\documentclass{article}

\usepackage{amsmath, amsfonts, amssymb}
\usepackage{booktabs}
\usepackage{enumerate}
\usepackage{hyperref}
\usepackage{multirow}
\usepackage{bm}
\usepackage{amsthm}
\usepackage[square, numbers, sort&compress]{natbib}

\newcommand{\dif}{\ensuremath{\,\mathrm{d}}}
\DeclareMathOperator{\E}{\mathbb{E}}
\newcommand{\indicator}{\ensuremath{\mathbb{I}}}
\newtheorem{theorem}{Theorem}

\begin{document}

\title{Variable Selection for Infectious Disease Transmission via Penalized Pairwise Accelerated Failure Time Models}

\author{Jeongjin Lee$^{1}$, Patrick M.\ Schnell$^{1}$, Grzegorz A.\ Rempala$^{1}$,\\
and Eben Kenah$^{1,*}$}

\date{\small
$^{1}$ Biostatistics Division, College of Public Health, The Ohio State University\\
$^{*}$ Correspondence: \texttt{kenah.1@osu.edu}
}

\maketitle
\sloppy

\begin{abstract}
In infectious disease transmission, regression and variable selection for time-to-event outcomes are complicated by transmission-induced dependence that violates standard independence assumptions.
Pairwise survival analysis addresses this dependence by modeling \emph{contact intervals} in ordered pairs, defined as the time from the onset of infectiousness in the pair's first individual to their first infectious contact with the second individual.
We propose a penalized pairwise accelerated failure time model for contact intervals that accommodates simultaneous covariate effects on infectiousness and susceptibility while jointly representing within-population transmission and infection from external sources.
For variable selection with many candidate covariates, we develop adaptive LASSO penalized maximum likelihood estimation with a coordinate descent algorithm, treating distributional baseline parameters as unpenalized nuisance parameters.
We establish selection consistency and asymptotic normality of the adaptive LASSO estimator under profile-likelihood regularity conditions in a local asymptotic regime in which information accumulates over all ordered pairs that are ever at risk of transmission during follow-up.
Simulation studies show accurate variable selection and improved estimation accuracy compared with the unpenalized estimator and the LASSO.
We illustrate the method using 2009 household influenza A (H1N1) surveillance data from Los Angeles County, where the adaptive LASSO produces a sparse, interpretable model and improves the precision of selected covariate effect estimates, including the effect of antiviral prophylaxis on susceptibility.

\end{abstract}

\section{Introduction}
A central goal in infectious disease epidemiology is to quantify how individual characteristics and interventions, such as age or vaccination status, influence transmission risk through their effects on infectiousness and susceptibility.
Household and other close-contact studies are especially informative because they directly observe clusters of intense exposure within a clearly-defined group of individuals at risk of infection and pairs at risk of transmission \citep{frost1938familial, fox1974family, cauchemez2009household, goeyvaerts2018household}.
Such studies often measure a large number of candidate covariates, so variable selection is essential for deriving parsimonious, interpretable models and identifying covariates most strongly associated with infectiousness and susceptibility.
Importantly, variable selection in this setting must account for transmission-driven dependence, rather than relying on independence assumptions that are violated in close-contact groups.
Infection outcomes are often analyzed using standard regression methods, for example treating secondary infection as a binary endpoint or modeling the time to infection in individuals.
However, infections within a close-contact group are not independent because one person’s infection changes the hazard of infection faced by others.
Leaving this dependence out of the analytical model can complicate interpretation and undermine both point estimation and uncertainty quantification, even when variance estimates are adjusted for clustering \citep{longini1988statistical, halloran1991study, sharker2021estimating}.
Infections acquired from outside the observed population can be nonnegligible in household studies, so failing to account for external infection can bias covariate effect estimates and reduce precision \citep{longini1982estimating, sharker2024pairwise}.
Together, these features motivate transmission models that represent within-group infection dynamics and external infection, coupled with principled variable selection methods to identify covariates associated with infectiousness and susceptibility without unrealistic independence assumptions.

Pairwise survival analysis offers a principled approach to transmission regression by modeling \emph{contact intervals} in ordered pairs, where the contact interval in the pair $ij$ is defined as the time from the onset of infectiousness in $i$ to the first infectious contact from $i$ to $j$.
The hazard function of the contact interval distribution represents instantaneous infectiousness, and its survival function can be used to calculate transmission probabilities (i.e., secondary attack risks).
Pairwise likelihoods can accommodate transmission-induced dependence while retaining a regression structure for covariate effects on infectiousness and susceptibility as well as pair- or group-level covariate effects \citep{kenah2011contact, kenah2013non, kenah2015semiparametric, sharker2024pairwise}.
\citet{kenah2013non} developed nonparametric contact-interval methods that combine a Nelson--Aalen-type estimator with an EM algorithm to estimate the cumulative hazard of infectious contact when who-infects-whom is unobserved.
Building on this framework, \citet{kenah2015semiparametric} developed semiparametric relative-risk regression for transmission within the observed population.
These methods did not represent infection from external sources as a separate hazard component.

Because infections acquired outside the observed population may otherwise be incorrectly attributed to within-population transmission, ignoring external infection is a fundamental limitation that can bias inference and limit the practical utility of pairwise survival analysis.
To address this, \citet{sharker2024pairwise} proposed pairwise accelerated failure time (AFT) regression models for contact intervals that explicitly incorporate external infection alongside within-population transmission.
Despite these advances, most existing work emphasizes estimation in relatively low-dimensional covariate settings, and variable selection procedures for pairwise survival analysis that are both computationally practical and theoretically grounded do not exist.

In this paper, we propose a penalized pairwise accelerated failure time model for contact intervals that jointly represents within-population transmission and infection from external sources.
Variable selection is performed by maximizing the adaptive LASSO (ALASSO) penalized likelihood \citep{zou2006adaptive}.
Estimation is carried out via coordinate descent, updating penalized regression coefficients using soft-thresholding while treating baseline and shape parameters as unpenalized nuisance parameters.
We establish selection consistency and asymptotic normality of the ALASSO estimator under a local asymptotic framework in which information accumulates through the ordered pairs that are ever at risk during follow-up.
Through simulation studies, we evaluate variable selection, estimation accuracy, and confidence interval calibration across several epidemiologic designs, including contact tracing with delayed entry and settings where who infected whom is unobserved.
Finally, we apply the method to household influenza A (H1N1) surveillance data from Los Angeles County in 2009, illustrating how ALASSO produces sparse, interpretable models and improves precision for selected covariate effects.

Section~\ref{sec:data} of this paper introduces the data structure and observation processes.
Section~\ref{sec:methods} presents the pairwise AFT likelihood for internal and external contact intervals and the adaptive LASSO estimation.
Section~\ref{sec:asymptotic} establishes asymptotic properties, and Section~\ref{sec:var_tune} describes variance approximation and tuning parameter selection.
Section~\ref{sec:simulation} reports simulation results, and Section~\ref{sec:application} presents the influenza application.

% ===============================
% Data
% ===============================
\section{Infectious Disease Data}
\label{sec:data}

\subsection{Stochastic S(E)IR Models}
The study population consists of $n$ individuals observed from time $0$ until a stopping time $T$.
For each individual $i \in \{1,\ldots,n\}$, we observe the infection time $t_i$, the onset of infectiousness at $t_i+\varepsilon_i$, and the end of infectiousness at $t_i+\varepsilon_i+\iota_i$, where $\varepsilon_i$ and $\iota_i$ denote the latent and infectious periods, respectively.
Each individual follows a stochastic SEIR-type process in which both durations have finite means.
At any time $t$, individual $i$ occupies one of four states: susceptible (S), exposed (E), infectious (I), or removed (R).
Transitions between these states determine the infection and recovery events used in the analysis.
The SEIR model simplifies to an SIR model when individuals have no latent period, so $\varepsilon_i = 0$.
We assume the latent period and infectious periods have finite mean and variance at all combinations of covariate values or that we have an S(E)I model where infected individuals do not recover.

An internal infection occurs when an individual is infected by another member of the study population.
For each ordered pair $(i,j)$, $C_{ij}$ indicates whether infectious contact from $i$ to $j$ would be possible if $i$ were infectious, where an infectious contact is defined to be sufficient to infect the recipient if they are susceptible.
$C_{ij}$ may represent an entry in the adjacency matrix of a (possibly directed or semidirected) contact network and need not be symmetric.

After becoming infectious at time $t_i + \varepsilon_i$, individual $i$ may produce an infectious contact directed toward another individual $j \ne i$.
If such a contact occurs, the time of the first infectious contact is
\[
t_{ij} = t_i + \varepsilon_i + \tau_{ij}^*,
\]
where $\tau_{ij}^*$ is the infectious contact interval.
The random variable $\tau_{ij}^*$ is strictly positive, and we set $\tau_{ij}^* = \infty$ when no infectious contact takes place.
Because infectious contact can occur only while $i$ is infectious, the interval must satisfy
\[
\tau_{ij}^* \in (0, \iota_i] \quad\text{or}\quad \tau_{ij}^* = \infty.
\]
The infectious contact interval $\tau_{ij}^*$ is constructed by drawing a contact interval $\tau_{ij}$ from a failure time distribution with hazard $h_{ij}(\tau)$.
If $\tau_{ij}\le \iota_i$ and $C_{ij}=1$, then $\tau_{ij}^*=\tau_{ij}$; otherwise $\tau_{ij}^*=\infty$.
The function $h_{ij}(\tau)$ describes the instantaneous infectiousness of $i$ at infectious age $\tau$, with cumulative hazard $H_{ij}(\tau)=\int_0^{\tau} h_{ij}(u)du$ and survival $S_{ij}(\tau)=\exp\{-H_{ij}(\tau)\}$.
The secondary attack risk (SAR), defined as the probability of infectious contact with $j$ during the entire infectious period of $i$, is $1-S_{ij}(\iota_i)$.

An external infection occurs when an individual is infected from outside the observed population.
Let $C_{0j}$ indicate whether $j$ is at risk of external infectious contact.
The external infectious contact time $\tau_{0j}^*$ is the first time $j$ receives infectious contact from an external source, with $\tau_{0j}^*=\infty$ if no such contact occurs.
We generate $\tau_{0j}^*$ by drawing an external contact interval $\tau_{0j}$ from a distribution with hazard $h_{0j}(t)$.
If $C_{0j}=1$, then $\tau_{0j}^*=\tau_{0j}$; otherwise $\tau_{0j}^*=\infty$.
For simplicity, we assume the external source is always infectious, with onset at time $0$ and an infinite infectious period, though this assumption may be relaxed by specifying its own onset and infectious duration.

For each individual $j$, let $\nu_j$ index the source of infection for $j$, taking $\nu_j=0$ for infection from outside the observed group and $\nu_j=\infty$ if $j$ is never infected.
When $\nu_j$ is known for all infected individuals, we say that \emph{who infected whom is observed}.
This may occur in a small, intensively monitored outbreak with near-complete case ascertainment and whole-genome viral sequencing that uniquely resolves transmission links.
Otherwise, we say that \emph{who infected whom is not observed}, even if $\nu_j$ is observed for only a subset of infected individuals.
The \emph{exposure set} of individual $j$ is the set of all sources of infection to whom $j$ was exposed while susceptible:
\[
\mathcal{W}_j
=
\bigl\{ i < \infty : (t_i+\varepsilon_i < t_j \ \text{or}\ i=0)\ \text{and}\ C_{ij}=1 \bigr\}.
\]
The \emph{infectious set} of individual $j$ is the set of all possible infectors $\nu_j$, which we denote $\mathcal{V}_j$. 
Given $C_{ij}$ and the transfer times between the states $(S,E,I,R)$, we have
\[
\mathcal{V}_j
\subseteq
\bigl\{ i < \infty : (t_i+\varepsilon_i < t_j \le t_i+\varepsilon_i+\iota_i \ \text{or}\ i=0)\ \text{and}\ C_{ij}=1 \bigr\}.
\]
If the infector of $j$ is known, then $\mathcal{V}_j=\{\nu_j\}$.
If $j$ remains uninfected, then $\mathcal{V}_j=\emptyset$.

\subsection{Observation}
For each susceptible individual $j$, the observed data include the infection outcome and the period during which infectious contact from any source could occur.
A contact interval $\tau_{ij}$ can be observed only if $j$ is infected by $i$, which requires $C_{ij}=1$ and that transmission from $i$ to $j$ is possible at the the infection time of $j$.
Contact intervals may be right-censored by the end of infectiousness in $i$, infection of $j$ from another source, or the end of follow-up.
For $i\neq 0$, let
\[
\mathcal{I}_i(t)
= \indicator\{ t_i+\varepsilon_i < t \le t_i+\varepsilon_i+\iota_i \},
\]
indicate whether $i$ is infectious at time $t$, and let $\mathcal{I}_0(t)$ denote whether external infectious contact is possible.
Let
\[
\mathcal{S}_j(t)=\indicator\{ t \le t_j \},
\qquad
\mathcal{O}(t)=\indicator\{ t \le T \}
\]
denote susceptibility of $j$ and whether observation is ongoing.
Because $\mathcal{I}_i(t)$, $\mathcal{S}_j(t)$, and $\mathcal{O}(t)$ are left-continuous, the at-risk process
\[
Y_{ij}(t)= C_{ij}\,\mathcal{I}_i(t)\,\mathcal{S}_j(t)\,\mathcal{O}(t)
\]
is also left-continuous and equals one whenever the pair $(i,j)$ is at risk of an observed infectious contact at time $t$.

% ===============================
% Methods
% ===============================
\section{Methods}
\label{sec:methods}

\subsection{Contact intervals and rate parameters}

Let $\lambda_{ij}$ denote the rate parameter of the contact interval distribution for pair $(i,j)$.
For internal pairs $(i,j)$ with $i\neq0$,
\begin{equation}
    \lambda_{ij}
    = \exp\bigl(\bm{x}_{ij}^\top \bm{\beta}\bigr)\,\lambda_0,
    \label{eq:internal}
\end{equation}
where $\lambda_0>0$ is the baseline internal transmission rate, and $\bm{x}_{ij}\in\mathbb{R}^p$ is the vector of pairwise covariates characterizing the infectiousness of individual $i$ and the susceptibility of individual $j$.
For external infections $(i=0)$,
\begin{equation}
    \lambda_{0j}
    = \exp\bigl(\bm{x}_{0j}^\top \bm{\beta}\bigr)\,\mu_0,
    \label{eq:external}
\end{equation}
where $\mu_0>0$ is the baseline external rate parameter and $\bm{x}_{0j}$ may include susceptibility or environmental covariates specific to $j$.
Both cases can be represented jointly as
\begin{equation}
    \lambda_{ij}
    =
    \exp\bigl(\bm{x}_{ij}^\top \bm{\beta}\bigr)
    \lambda_0^{1-\indicator \{i=0\}}
    \mu_0^{\indicator \{i=0\}},
    \label{eq:lambdaij}
\end{equation}
where $\indicator \{i=0\}$ is an indicator for external exposure.
In subsequent penalized estimation, we penalize only the regression coefficients $\bm{\beta}$.
Baseline rate parameters $(\lambda_0,\mu_0)$ and any distributional shape parameters are treated as nuisance parameters and are not penalized.

\subsection{Pairwise AFT likelihood}

We formulate a parametric pairwise likelihood for contact interval data using the counting-process representation developed in earlier pairwise survival work \citep{kenah2011contact, kenah2013non, sharker2024pairwise}.
Let $\bm{\theta}=(\bm{\beta}^\top,\bm{\psi}^\top)^\top$ denote the full parameter vector, where $\bm{\beta}$ collects regression coefficients for covariate effects and $\bm{\psi}$ collects the remaining nuisance parameters indexing the chosen contact-interval family (e.g., baseline log-rate terms and, when applicable, log-shape terms for internal and external sources).
For each ordered pair $(i,j)$, let $h_{ij}(\cdot;\bm{\theta})$ and $S_{ij}(\cdot;\bm{\theta})$ denote the hazard and survival functions of the contact interval distribution parameterized by the pair-specific rate $\lambda_{ij}$ in \eqref{eq:lambdaij}.

\subsubsection{Who-infected-whom is observed}

When the infector of each infected individual is observed, infectious contact from $i$ to $j$ may be represented by a counting process $N_{ij}(t)$ that jumps at the time of infectious contact from $i$ to $j$.
Using the at-risk process $Y_{ij}(t)$ defined in Section~\ref{sec:data}, the log likelihood contribution for an ordered pair $(i,j)$ over $[0,T]$ is
\begin{equation}
\ell_{ij}(\bm{\theta})
=
\int_0^T
  \log h_{ij}\big(u - t_i - \varepsilon_i;\bm{\theta}\big)\,
  \dif N_{ij}(u)
-
\int_0^T
  h_{ij}\big(u - t_i - \varepsilon_i;\bm{\theta}\big)\,
  Y_{ij}(u)\,
  \dif u,
\label{eq:ll_ij_obs}
\end{equation}
with the convention that for the external source $i=0$ the infectious age is $u$, so $h_{0j}(u;\bm{\theta})$ is used without a time shift.
The full log likelihood is obtained by summing \eqref{eq:ll_ij_obs} over all ordered pairs with $i\in\{0,1,\ldots,n\}$ and $i\neq j$:
\begin{equation}
\ell(\bm{\theta})
=
\sum_{j=1}^{n}
\sum_{\substack{i=0\\ i\neq j}}^{n}
\ell_{ij}(\bm{\theta}).
\label{eq:ll_obs}
\end{equation}
In this setting, exactly one process $N_{ij}(t)$ jumps for each infected $j$ (the observed infector) and contributes a hazard term.
Each ordered pair $(i, j)$ contributes to the log likelihood through its cumulative hazard term.
The cumulative hazard terms in the log likelihood do not depend on who-infected-whom.

\subsubsection{Who-infected-whom is not observed}

In the more realistic case when who-infected-whom is not observed, we do not observe the individual pair processes $\{N_{ij}(t)\}$, but only the aggregated infection counting process
\[
N_{\cdot j}(t)
=
\sum_{\substack{i=0\\ i\neq j}}^{n} N_{ij}(t),
\]
which records the time at which individual $j$ becomes infected, regardless of source.
In this case, the total hazard for individual $j$ at time $t$ is obtained by summing over all potential sources of infectious contact,
\begin{equation}
h_{\cdot j}(t;\bm{\theta})
=
\sum_{\substack{i=0\\ i\neq j}}^{n}
C_{ij}\,\mathcal{I}_i(t)\,
h_{ij}\big(t - t_i - \varepsilon_i;\bm{\theta}\big),
\label{eq:totalhaz}
\end{equation}
again with the convention that $h_{0j}(t;\bm{\theta})$ uses infectious age $t$ for $i=0$.
The corresponding log likelihood contribution for individual $j$ is
\begin{equation}
\ell_{\cdot j}(\bm{\theta})
=
\int_0^T
  \log h_{\cdot j}(u;\bm{\theta})\,
  \dif N_{\cdot j}(u)
-
\int_0^T
  h_{\cdot j}(u;\bm{\theta})\,
  \mathcal{S}_j(u)\,
  \dif u,
\label{eq:ll_j}
\end{equation}
and the observed-data log likelihood is
\begin{equation}
\ell(\bm{\theta})
=
\sum_{j=1}^{n} \ell_{\cdot j}(\bm{\theta}).
\label{eq:ll_unobs}
\end{equation}
If who-infected-whom is partially observed, we have
\begin{equation}
h_{\cdot j}(t;\bm{\theta})
=
\sum_{\substack{i\in \mathcal{V}_j}}^{n}
C_{ij}\,\mathcal{I}_i(t)\,
h_{ij}\big(t - t_i - \varepsilon_i;\bm{\theta}\big).
\end{equation}
where the infectious set $\mathcal{V}_j$ excludes any $i$ who is known not to have infected $j$.

\subsection{Adaptive LASSO Estimation}
\label{sec:alasso}

We estimate the full parameter vector $\bm\theta$ defined above by maximizing the adaptive $\ell_{1}$-penalized log likelihood
\begin{equation}
Q_{\mathrm{AL}}(\bm{\theta})
=
\ell(\bm{\theta})
-
\lambda
\sum_{k=1}^{p}
w_k |\beta_k|,
\qquad
w_k = |\widetilde{\beta}_k|^{-\eta},
\label{eq:penll}
\end{equation}
where $p=\dim(\bm\beta)$ is the number of penalized regression coefficients, $\lambda>0$ controls the penalty strength, and $\eta>1$ (typically $\eta=1.75$) is the adaptive exponent.
The weights are formed from an initial consistent estimator $\widetilde{\bm{\beta}}$ (e.g., the unpenalized maximum likelihood estimator \citep{sharker2024pairwise}), so that coefficients estimated to be near zero receive larger penalties.
When \emph{who-infected-whom is observed}, we take $\ell(\bm{\theta})$ to be the complete-data log likelihood in \eqref{eq:ll_obs}; when \emph{who-infected-whom is not observed}, we instead use the observed-data log likelihood in \eqref{eq:ll_unobs}.
In all penalized fits, we penalize only the regression coefficients $\bm\beta$ and treat baseline and shape parameters in $\bm\psi$ as unpenalized nuisance parameters.
Setting $w_k\equiv 1$ reduces \eqref{eq:penll} to the standard LASSO \citep{tibshirani1996regression}.

To optimize \eqref{eq:penll} for a fixed $\lambda$, we employ coordinate descent \citep{wright2015coordinate,Tong2008cd, breheny2011coordinate}.
Throughout, we index the full parameter vector by $j\in\{1,\ldots,d\}$, where $d=\dim(\bm\theta)$, and index penalized regression coefficients by $k\in\{1,\ldots,p\}$.
We order $\bm\theta$ as $\bm\theta=(\bm\beta^\top,\bm\psi^\top)^\top$, so that $\theta_k=\beta_k$ for $k=1,\ldots,p$.
Let $f(\bm\theta)=-\ell(\bm\theta)$ denote the unpenalized negative log likelihood.
For coordinate $j$, define the one-dimensional slice
\[
f_j(z)=f(\theta_1,\ldots,\theta_{j-1},z,\theta_{j+1},\ldots).
\]
We approximate the local gradient and curvature at the current value $\theta_j$ using centered finite differences,
\begin{equation}
g_j
=
\frac{f_j(\theta_j+\varepsilon_{\mathrm{fd}})-f_j(\theta_j-\varepsilon_{\mathrm{fd}})}{2\varepsilon_{\mathrm{fd}}},
\qquad
H_{jj}
=
\frac{f_j(\theta_j+\varepsilon_{\mathrm{fd}})-2f_j(\theta_j)+f_j(\theta_j-\varepsilon_{\mathrm{fd}})}{\varepsilon_{\mathrm{fd}}^{2}},
\label{eq:fd_gh}
\end{equation}
with $\varepsilon_{\mathrm{fd}}=10^{-3}$, and enforce $H_{jj}\ge h_{\min}$ with $h_{\min}=10^{-3}$.
For a penalized regression coefficient $\beta_k$, a second order approximation yields the weighted soft thresholding update
\begin{equation}
\beta_k^{\mathrm{new}}
=
S\left(
  \beta_k - \frac{g_k}{H_{kk}},
  \frac{\lambda w_k}{H_{kk}}
\right),
\qquad
S(z,\tau)
=
\operatorname{sign}(z)\,\max\bigl(|z|-\tau,0\bigr),
\label{eq:update}
\end{equation}
whereas for an unpenalized nuisance coordinate $\theta_j\in\bm\psi$ the update reduces to the Newton step
\begin{equation}
\theta_j^{\mathrm{new}}
=
\theta_j
-
\frac{g_j}{H_{jj}}.
\label{eq:update_psi}
\end{equation}
A derivation of \eqref{eq:update} is provided in the Appendix, and the algorithm is summarized below.

\begin{enumerate}[(1)]
\item \textbf{Inputs.}
Log likelihood $\ell(\bm\theta)$; dimensions $d=\dim(\bm\theta)$ and $p=\dim(\bm\beta)$; initial value $\bm\theta^{(0)}$; penalty $\lambda$; adaptive weights $\bm w=(w_1,\ldots,w_{p})^\top$; finite difference step $\varepsilon_{\mathrm{fd}}=10^{-3}$; curvature floor $h_{\min}=10^{-3}$; step cap $\Delta_{\max}=1$; tolerance $\mathrm{tol}=10^{-3}$; maximum sweeps $R_{\max}=50$.

\item \textbf{Initialize.}
Set $\bm\theta \leftarrow \bm\theta^{(0)}$.

\item \textbf{For $r=1,\ldots,R_{\max}$:}
\begin{enumerate}[(a)]
\item Set $\Delta_{\infty}\leftarrow 0$.

\item \textbf{Update penalized coordinates ($k=1,\ldots,p$).}
\begin{enumerate}[(i)]
\item Compute $(g_k,H_{kk})$ by \eqref{eq:fd_gh} at the current $\theta_k=\beta_k$, then set $H_{kk}\leftarrow \max(H_{kk},h_{\min})$.

\item Compute
\[
\beta_k^{\mathrm{new}}
=
S\!\left(\beta_k-\frac{g_k}{H_{kk}},\ \frac{\lambda w_k}{H_{kk}}\right),
\qquad
\Delta \leftarrow \beta_k^{\mathrm{new}}-\beta_k.
\]

\item Apply the step cap $\Delta\leftarrow \operatorname{sign}(\Delta)\min(|\Delta|,\Delta_{\max})$, update $\beta_k\leftarrow \beta_k+\Delta$ (equivalently $\theta_k\leftarrow \theta_k+\Delta$), and set $\Delta_{\infty}\leftarrow \max(\Delta_{\infty},|\Delta|)$.
\end{enumerate}

\item \textbf{Update unpenalized coordinates ($j=p+1,\ldots,d$).}
\begin{enumerate}[(i)]
\item Compute $(g_j,H_{jj})$ by \eqref{eq:fd_gh} at the current $\theta_j$, then set $H_{jj}\leftarrow \max(H_{jj},h_{\min})$.

\item Compute
\[
\theta_j^{\mathrm{new}}=\theta_j-\frac{g_j}{H_{jj}},
\qquad
\Delta \leftarrow \theta_j^{\mathrm{new}}-\theta_j.
\]

\item Apply the step cap $\Delta\leftarrow \operatorname{sign}(\Delta)\min(|\Delta|,\Delta_{\max})$, update $\theta_j\leftarrow \theta_j+\Delta$, and set $\Delta_{\infty}\leftarrow \max(\Delta_{\infty},|\Delta|)$.
\end{enumerate}

\item \textbf{Check convergence.}
If $\Delta_{\infty}<\mathrm{tol}$, stop and return $\widehat{\bm\theta}_\lambda=\bm\theta$.
\end{enumerate}

\item \textbf{Return.}
If convergence is not reached after $R_{\max}$ sweeps, return the final iterate $\widehat{\bm\theta}_\lambda=\bm\theta$.
\end{enumerate}

\subsection{Asymptotic properties}
\label{sec:asymptotic}

We establish the oracle properties of the adaptive LASSO only in the setting where \emph{who-infected-whom is observed}.
We use a local asymptotic framework in which information accumulates through the \emph{ordered pairs that are ever at risk of transmission} during follow-up.
For each data set in an increasing sequence of observed-WIW data sets, define
\begin{equation}
\mathcal{R}
=
\Bigl\{(i,j): \exists\,\tau \in [0,T]\ \text{such that}\ Y_{ij}(\tau)=1 \Bigr\},
\qquad
m = |\mathcal{R}|.
\label{eq:def_R_m}
\end{equation}
All limits below are taken along this sequence with $m\to\infty$ in probability.
Let $\ell_m(\bm\theta)$ denote the complete-data log likelihood in \eqref{eq:ll_obs} for an observed-WIW data set with $m$ ordered pairs ever at risk.
The pair contributions in \eqref{eq:ll_obs} need not be independent; epidemic dependence is accommodated through the profile-likelihood regularity conditions below.

Let $\bm{\theta}_0=(\bm{\beta}_0^\top,\bm{\psi}_0^\top)^\top$ be the true interior parameter, and let $\mathcal{S}_0=\{k:\beta_{0,k}\neq 0\}$.
Because $\bm\psi$ is unpenalized, we work with the profile log likelihood
\[
\ell_{\text{profile},m}(\bm{\beta})
=
\sup_{\bm{\psi}}\ell_m(\bm{\beta},\bm{\psi}),
\]
and assume that the supremum is attained locally by a regular nuisance-parameter estimator.

As in \citet{zhang2007adaptive}, we assume standard likelihood regularity conditions.
For clarity, define the profile score vector, its $k$th component, and the profile Hessian by
\begin{align*}
\dot\ell_{\text{profile},m}(\bm\beta)
&=
\frac{\partial\ell_{\text{profile},m}(\bm\beta)}
{\partial\bm\beta},
&
\dot\ell_{\text{profile},m,k}(\bm\beta)
&=
\frac{\partial\ell_{\text{profile},m}(\bm\beta)}
{\partial\beta_k},
\\
\ddot\ell_{\text{profile},m}(\bm\beta)
&=
\frac{\partial^2\ell_{\text{profile},m}(\bm\beta)}
{\partial\bm\beta\,\partial\bm\beta^\top}.
\end{align*}
The only consequences of the regularity conditions used in the proofs are that $\ell_{\text{profile},m}$ is twice continuously differentiable near $\bm\beta_0$ and, for every fixed $C<\infty$,
\begin{align}
\bm S_m
\equiv
\frac{1}{\sqrt{m}}\dot\ell_{\text{profile},m}(\bm\beta_0)
&\to_d N(\bm 0,\bm\Sigma),
\label{eq:profile_score}\\[-2pt]
-\frac{1}{m}\ddot\ell_{\text{profile},m}(\bm\beta)
&=
\bm\Sigma+o_p(1)
\quad\text{uniformly for }
\|\bm\beta-\bm\beta_0\|\le C/\sqrt m.
\label{eq:profile_information}
\end{align}
The $k$th component of the normalized profile score is therefore
\[
S_{m,k}
=
[\bm S_m]_k
=
\frac{1}{\sqrt m}
\dot\ell_{\text{profile},m,k}(\bm\beta_0).
\]
Here, $\bm\Sigma$ is the positive-definite limiting efficient information for $\bm\beta$.

Let $\widetilde{\bm\beta}_m$ be an initial estimator, such as the unpenalized maximum likelihood estimator of \citet{sharker2024pairwise}.
For a fixed $\eta>1$, define the weights $w_{m,k}=|\widetilde\beta_{m,k}|^{-\eta}$.
For a tuning sequence $\lambda_m>0$, consider the profiled penalized criterion
\begin{equation}
\mathcal Q_m(\bm\beta)
=
\ell_{\text{profile},m}(\bm\beta)
-
\lambda_m\sum_{k=1}^p w_{m,k}|\beta_k|.
\label{eq:profile_penll}
\end{equation}

\begin{theorem}
\label{thm:rootm}
Suppose that:
\begin{enumerate}[(i)]
\item The dimensions $p$ and $\dim(\bm\psi)$ are fixed, and $\bm\beta_0$ is an interior point.

\item The profile log likelihood $\ell_{\text{profile},m}$ is twice continuously differentiable near $\bm\beta_0$, and \eqref{eq:profile_score}--\eqref{eq:profile_information} hold with $\bm\Sigma$ positive definite.

\item The initial estimator is root-$m$ consistent:
\[
\|\widetilde{\bm\beta}_m-\bm\beta_0\|=O_p(m^{-1/2}).
\]

\item The tuning parameter satisfies
\[
\frac{\lambda_m}{\sqrt m}=O(1).
\]
\end{enumerate}
Then there exists a sequence of local maximizers $\widehat{\bm\beta}_{\lambda,m}$ of $\mathcal Q_m$ such that
\[
\bigl\|\widehat{\bm\beta}_{\lambda,m}-\bm\beta_0\bigr\|
=O_p(m^{-1/2}).
\]
\end{theorem}

Theorem~\ref{thm:rootm} is a localization result.
It guarantees that the penalized criterion has a local maximizer within an $m^{-1/2}$ neighborhood of the true coefficient vector, but it does not yet claim that zero coefficients are selected exactly.

\begin{theorem}
\label{thm:oracle}
Suppose that:
\begin{enumerate}[(i)]
\item Assumptions (i)--(iii) of Theorem~\ref{thm:rootm} hold.

\item The tuning parameter satisfies
\[
\frac{\lambda_m}{\sqrt{m}}\to 0,
\qquad
\lambda_m m^{(\eta-1)/2}\to\infty.
\]
\end{enumerate}
Let $\widehat{\bm\beta}_{\lambda,m}$ be any root-$m$ consistent local maximizer of $\mathcal Q_m$.
Then
\begin{align*}
\Pr\left(
  \widehat{\bm\beta}_{\lambda,m,\mathcal{S}_0^c}=\bm 0
\right)
&\to 1,\\[4pt]
\sqrt{m}\Bigl(
  \widehat{\bm{\beta}}_{\lambda,m,\mathcal{S}_0}
  -
  \bm{\beta}_{0,\mathcal{S}_0}
\Bigr)
&\to_d
N\left(\bm 0,\,\bm{\Sigma}_{\mathcal{S}_0}^{-1}\right),
\end{align*}
where $\bm{\Sigma}_{\mathcal{S}_0}$ is the principal submatrix of $\bm{\Sigma}$ indexed by $\mathcal{S}_0$.
\end{theorem}

Theorem~\ref{thm:oracle} strengthens this local consistency result.
The first conclusion gives consistent variable selection, and the second shows that the estimated nonzero coefficients have the same first-order limit distribution as they would if the true active set were known in advance.

\subsection{Variance Estimation and Tuning Parameter Selection}
\label{sec:var_tune}

\subsubsection{Approximate variance}
\label{sec:alasso_var}

Following the standard error formulations of \citet{fan2001variable, zou2006adaptive} and the extension to survival models by \citet{zhang2007adaptive}, we approximate the adaptive LASSO estimator by a local quadratic (iteratively reweighted ridge) approximation to the $\ell_1$ penalty.
Let $\widehat{\bm\theta}_\lambda=(\widehat{\bm\beta}_\lambda^\top,\widehat{\bm\psi}_\lambda^\top)^\top$ maximize \eqref{eq:penll} for a fixed $\lambda$, and let
\[
\bm U(\bm\theta)=\frac{\partial \ell(\bm\theta)}{\partial \bm\theta},
\qquad
\bm H(\bm\theta)=-\frac{\partial^2 \ell(\bm\theta)}{\partial \bm\theta\,\partial \bm\theta^\top}
\]
denote the score and observed information matrix.

For numerical stability, fix a small constant $\xi = 10^{-5}$.
At a current iterate $\bm\theta^{(r)}$, we use the local quadratic approximation
\[
|\beta_k|
\approx
\frac{\beta_k^2}{2\left(|\beta_k^{(r)}|+\xi\right)},
\]
so that the adaptive penalty $\lambda\sum_{k=1}^p w_k|\beta_k|$ is approximated by $\frac{\lambda}{2}\bm\theta^\top \bm A_{(r)} \bm\theta$, where $\bm A_{(r)}$ is block-diagonal with a nonzero block only on the penalized $\bm\beta$ coordinates:
\[
\bm A_{(r)}
=
\begin{pmatrix}
\mathrm{diag}\left(
\frac{w_1}{|\beta^{(r)}_1|+\xi},\ldots,
\frac{w_p}{|\beta^{(r)}_p|+\xi}
\right) & \bm 0 \\
\bm 0 & \bm 0
\end{pmatrix}.
\]
In particular, all coordinates in $\bm\psi$ (e.g., baseline log-rate and log-shape parameters) correspond to zeros in $\bm A_{(r)}$ and are therefore unpenalized.
Using the same approximation, we estimate the covariance of $\widehat{\bm\theta}_\lambda$ by the sandwich form
\begin{equation}
\widehat{\mathrm{cov}}(\widehat{\bm\theta}_\lambda)
=
\left\{\bm H(\widehat{\bm\theta}_\lambda)+\lambda \bm A_{\lambda}\right\}^{-1}
\,
\bm H(\widehat{\bm\theta}_\lambda)
\,
\left\{\bm H(\widehat{\bm\theta}_\lambda)+\lambda \bm A_{\lambda}\right\}^{-1},
\label{eq:alasso_cov_theta}
\end{equation}
where $\bm A_{\lambda}$ is $\bm A_{(r)}$ evaluated at the final iterate.
Wald-type standard errors are obtained from the diagonal entries of \eqref{eq:alasso_cov_theta}.

In practice, inference for penalized regression effects is reported only for the nonzero coefficients at the selected tuning parameter $\lambda$.
Let $\mathcal S=\{k:\widehat{\beta}_{\lambda,k}\neq 0\}$ denote the active penalized set and let $\mathcal U$ denote the index set of all unpenalized coordinates (including $\bm\psi$ and any unpenalized regression components such as intercept terms).
We apply \eqref{eq:alasso_cov_theta} to the reduced parameter vector $\widehat{\bm\theta}_{\lambda,\mathcal A}$ with $\mathcal A=\mathcal U\cup\mathcal S$ using the corresponding $\mathcal A\times\mathcal A$ submatrices of $\bm H(\widehat{\bm\theta}_\lambda)$ and $\bm A_\lambda$, treating coordinates in $\mathcal S^c$ as fixed at zero.
A detailed derivation is provided in Appendix.

\subsubsection{Tuning parameter selection}
For each candidate $\lambda$, we compute $\widehat{\bm{\theta}}_{\lambda} = (\widehat{\bm{\beta}}_{\lambda}^\top,\widehat{\bm{\psi}}_{\lambda}^\top)^\top$ and select $\lambda$ by minimizing an information criterion computed from the unpenalized log likelihood evaluated at the penalized estimate:
\begin{align}
\mathrm{AIC}(\lambda)
&=
2 \, \mathrm{Eff}(\lambda)-2 \, \ell (\widehat{\bm{\theta}}_{\lambda}), \label{eq:aic_lambda}\\
\mathrm{BIC}(\lambda)
&=
\log(m) \, \mathrm{Eff}(\lambda)-2 \, \ell (\widehat{\bm{\theta}}_{\lambda} ), \label{eq:bic_lambda}
\end{align}
where $m=|\mathcal{R}|$ is the number of ordered pairs ever at risk in \eqref{eq:def_R_m} and
\[
{\mathrm{Eff}}(\lambda)=\sum_{k=1}^{p}\indicator\left(\widehat{\beta}_{\lambda,k}\neq 0\right)
\]
is the effective number of selected covariate effects.
Because our primary goal is variable selection and interpretable transmission effects rather than prediction, we use BIC in \eqref{eq:bic_lambda} as the primary tuning criterion.

% ===============================
% Simulation
% ===============================
\section{Simulation Study}
\label{sec:simulation}

We conducted a simulation study to assess the finite-sample performance of the adaptive LASSO for pairwise AFT models under four epidemiologic study designs.
The simulations followed a household transmission model in which infection spreads within small, fully connected groups.

The population consisted of 500 households of size five (2{,}500 individuals).
Within each household, all ordered pairs of individuals were treated as potential transmission pairs, and no transmission was allowed between households.
Each infected individual had a latent period of 0 and an infectious period of length 1.
Epidemics were simulated until 1{,}000 infections occurred, ensuring adequate numbers of both internal and external infection events.

In all simulations, baseline internal contact intervals followed a log-logistic distribution with shape parameter $\gamma = 2$ and baseline rate $\lambda_0 = 0.5$, while external contact intervals followed an exponential distribution with rate $\mu_0 = -0.5 \ln 0.8$.
The rate parameter for the contact interval distribution for each pair $(i,j)$ was
\[
\lambda_{ij}
=
\exp\Bigl(
  \beta_{\mathrm{inf}}^{\top} \, X_i
  + \beta_{\mathrm{sus}}^{\top} \, X_j
  + \indicator \{i\neq 0\} \, \ln \lambda_0
  + \indicator \{i=0\} \, \ln \mu_0
\Bigr),
\qquad X_0 = 0.
\]
For covariates, we considered two scenarios:
\begin{itemize}
    \item[(1)] \textbf{Multivariate normal covariates.}
    Each individual $i$ was assigned a five-dimensional continuous covariate vector
    \[
    X_i \sim \mathcal{N}_5(0,\,\Sigma),
    \qquad
    \Sigma_{jk} =
    \begin{cases}
    1, & j=k,\\
    0.5, & j\ne k,
    \end{cases}
    \]
    and covariate effects on infectiousness and susceptibility were
    \[
    \beta_{\mathrm{inf}} = (0.7,\,0.7,\,0,\,0,\,0),
    \qquad
    \beta_{\mathrm{sus}} = (0.7,\,0.7,\,0,\,0,\,0).
    \]

    \item[(2)] \textbf{Bernoulli covariates.}
    In the second scenario, each individual $i$ was assigned a five-dimensional binary covariate vector
    \[
    X_i = (Z_{i1},\dots,Z_{i5})^{\top},
    \qquad
    Z_{ik} \sim \mathrm{Bernoulli}(0.5) \;\text{independently},
    \]
    and the covariate effects were kept the same as in the multivariate normal case.
\end{itemize}

Each simulated data set was analyzed under four study designs:
(i) complete cohort with who-infected-whom (WIW) observed, (ii) complete cohort with WIW unobserved, (iii) contact tracing (CT) with delayed entry and WIW observed, and (iv) contact tracing (CT) with delayed entry and WIW unobserved.
In the complete-cohort designs (i) and (ii), follow-up for all individuals begins at time zero, which serves as the time origin for external infectious contact intervals and implies no left truncation.
In the contact-tracing designs (iii) and (iv), follow-up of each individual begins at the infection time of the index case in their household, with time at risk of external infectious contact prior to the start of follow-up left-truncated and individuals in households with no infections excluded from the study.

For the LASSO and adaptive LASSO (ALASSO) estimators, we employed the coordinate descent algorithm described above, and the tuning parameter $\lambda$ was selected by minimizing the proposed BIC.
To construct the candidate set of tuning parameters, we generated a grid of 50 values evenly spaced on the log scale between 0.1 and 100:
\[
\lambda
=
\exp\left(
  \log(0.1)
  +
  \frac{g-1}{49}\,\bigl[\log(100)-\log(0.1)\bigr]
\right),
\qquad g = 1,\dots,50.
\]
The MLE, used to initialize the adaptive LASSO weights, was obtained by maximizing the pairwise likelihood of \citet{sharker2024pairwise} in R using TranStat (version 0.3.7; \url{https://github.com/ekenah/TranStat}).
The R functions implementing penalized regression (LASSO and ALASSO) are provided in the Supplementary Material.

Estimation accuracy was assessed using the sum squared error (SSE),
\[
\mathrm{SSE}
=
(\widehat{\beta}_{\mathrm{inf}} - \beta_{\mathrm{inf}})^{\top}
(\widehat{\beta}_{\mathrm{inf}} - \beta_{\mathrm{inf}})
+
(\widehat{\beta}_{\mathrm{sus}} - \beta_{\mathrm{sus}})^{\top}
(\widehat{\beta}_{\mathrm{sus}} - \beta_{\mathrm{sus}}),
\]
which measures the total squared deviation of the estimated infectiousness and susceptibility effects from their true values.
Across 100 Monte Carlo iterations, we report the mean squared error (MSE) as the average SSE across iterations.
Variable selection performance was evaluated using false positive (FP) and false negative (FN) counts.
A false positive corresponds to a covariate whose true coefficient is zero but is incorrectly estimated as nonzero, indicating selection of a noise variable, whereas a false negative occurs when a truly nonzero coefficient is incorrectly shrunk to zero, indicating failure to identify a true signal.
FP and FN were computed separately for infectiousness and susceptibility to reflect their distinct contributions in pairwise transmission models.
Together, these metrics provide a comprehensive evaluation of the MLE, LASSO, and ALASSO under both correctly specified and misspecified contact-interval models, across multiple covariate structures and epidemiologic designs.

\begin{table}[h!]
\scriptsize
\centering
\caption{
Performance of the MLE, LASSO, and ALASSO for data generated with log-logistic internal and exponential external contact intervals.
Covariates are five-dimensional multivariate normal vectors with $\Sigma_{jj}=1$ and $\Sigma_{jk}=0.5$ for $j\neq k$.
Results are based on 100 replicates.
}
\label{tab:sim_mvn}
\vspace{4pt}
\begin{tabular}{l l c c c c c}
\toprule
\multicolumn{7}{c}{\textbf{(A) Correctly specified model: internal = log-logistic, external = exponential}}\\
\midrule
Design & Method & MSE & FP$_{\text{sus}}$ & FN$_{\text{sus}}$ & FP$_{\text{inf}}$ & FN$_{\text{inf}}$ \\
\midrule
Full Cohort (WIW obs)   & MLE    & 0.0116 & --   & --   & --   & --   \\
                        & LASSO  & 0.0129 & 0.90 & 0.00 & 1.10 & 0.00 \\
                        & ALASSO & 0.0047 & 0.01 & 0.00 & 0.02 & 0.00 \\[2pt]
Full Cohort (WIW unobs) & MLE    & 0.0148 & --   & --   & --   & --   \\
                        & LASSO  & 0.0178 & 0.95 & 0.00 & 0.93 & 0.00 \\
                        & ALASSO & 0.0058 & 0.01 & 0.00 & 0.01 & 0.00 \\[2pt]
CT (WIW obs)            & MLE    & 0.0125 & --   & --   & --   & --   \\
                        & LASSO  & 0.0140 & 0.77 & 0.00 & 1.09 & 0.00 \\
                        & ALASSO & 0.0052 & 0.02 & 0.00 & 0.02 & 0.00 \\[2pt]
CT (WIW unobs)          & MLE    & 0.0162 & --   & --   & --   & --   \\
                        & LASSO  & 0.0251 & 0.83 & 0.00 & 0.84 & 0.00 \\
                        & ALASSO & 0.0068 & 0.01 & 0.00 & 0.02 & 0.00 \\
\midrule
\multicolumn{7}{c}{\textbf{(B) Misspecified model: internal = Weibull, external = exponential}}\\
\midrule
Design & Method & MSE & FP$_{\text{sus}}$ & FN$_{\text{sus}}$ & FP$_{\text{inf}}$ & FN$_{\text{inf}}$ \\
\midrule
Full Cohort (WIW obs)   & MLE    & 0.0183 & --   & --   & --   & --   \\
                        & LASSO  & 0.0215 & 1.17 & 0.00 & 1.39 & 0.00 \\
                        & ALASSO & 0.0078 & 0.03 & 0.00 & 0.04 & 0.00 \\[2pt]
Full Cohort (WIW unobs) & MLE    & 0.0236 & --   & --   & --   & --   \\
                        & LASSO  & 0.0270 & 1.15 & 0.00 & 1.29 & 0.00 \\
                        & ALASSO & 0.0102 & 0.03 & 0.00 & 0.05 & 0.00 \\[2pt]
CT (WIW obs)            & MLE    & 0.0202 & --   & --   & --   & --   \\
                        & LASSO  & 0.0245 & 1.07 & 0.00 & 1.41 & 0.00 \\
                        & ALASSO & 0.0091 & 0.05 & 0.00 & 0.06 & 0.00 \\[2pt]
CT (WIW unobs)          & MLE    & 0.0274 & --   & --   & --   & --   \\
                        & LASSO  & 0.0556 & 0.93 & 0.00 & 0.95 & 0.00 \\
                        & ALASSO & 0.0140 & 0.02 & 0.00 & 0.05 & 0.00 \\
\midrule
\multicolumn{7}{c}{\textbf{(C) Misspecified model: internal = exponential, external = exponential}}\\
\midrule
Design & Method & MSE & FP$_{\text{sus}}$ & FN$_{\text{sus}}$ & FP$_{\text{inf}}$ & FN$_{\text{inf}}$ \\
\midrule
Full Cohort (WIW obs)   & MLE    & 0.0617 & --   & --   & --   & --   \\
                        & LASSO  & 0.0221 & 0.88 & 0.00 & 0.94 & 0.00 \\
                        & ALASSO & 0.0496 & 0.00 & 0.00 & 0.01 & 0.00 \\[2pt]
Full Cohort (WIW unobs) & MLE    & 0.0418 & --   & --   & --   & --   \\
                        & LASSO  & 0.0131 & 1.03 & 0.00 & 0.91 & 0.00 \\
                        & ALASSO & 0.0261 & 0.00 & 0.00 & 0.01 & 0.00 \\[2pt]
CT (WIW obs)            & MLE    & 0.0775 & --   & --   & --   & --   \\
                        & LASSO  & 0.0294 & 0.70 & 0.00 & 1.00 & 0.00 \\
                        & ALASSO & 0.0640 & 0.01 & 0.00 & 0.02 & 0.00 \\[2pt]
CT (WIW unobs)          & MLE    & 0.0494 & --   & --   & --   & --   \\
                        & LASSO  & 0.0228 & 0.69 & 0.00 & 0.73 & 0.00 \\
                        & ALASSO & 0.0326 & 0.01 & 0.00 & 0.00 & 0.00 \\
\bottomrule
\end{tabular}
\end{table}

Table~\ref{tab:sim_mvn} summarizes estimation and variable-selection performance of the MLE, LASSO, and ALASSO across the four study designs under multivariate normal covariates.
Under correct specification (Panel A), the ALASSO achieved the smallest MSE in each design and selected sparse models with FP counts near zero, while retaining all truly nonzero effects (FN $=0$ for both infectiousness and susceptibility).
In contrast, the LASSO selected approximately one noise covariate on average in each component and yielded larger MSE than the ALASSO.
Under misspecification of the internal contact-interval distribution to a Weibull model (Panel B), MSE increased for all methods, with the largest degradation when who-infected-whom was unobserved.
Nevertheless, the ALASSO remained the most accurate estimator and continued to substantially reduce false positives relative to the LASSO.
Under more severe misspecification to an exponential internal model (Panel C), the MLE exhibited substantial loss of accuracy, whereas both penalized estimators continued to control false positives.
In this setting, the ALASSO exhibited larger MSE than in Panels A and B, whereas the LASSO showed a mixed pattern across designs, with MSE increasing in some settings but decreasing in others.

\begin{table}[h!]
\scriptsize
\centering
\caption{
Performance of the MLE, LASSO, and ALASSO for data generated with log-logistic internal and exponential external contact intervals.
Covariates are five-dimensional binary vectors $X_i=(Z_{i1},\dots,Z_{i5})^{\top}$ with $Z_{ik}\sim \mathrm{Bernoulli}(0.5)$ independently.
Results are based on 100 replicates.
}
\label{tab:sim_bern}
\vspace{4pt}
\begin{tabular}{l l c c c c c}
\toprule
\multicolumn{7}{c}{\textbf{(A) Correctly specified model: internal = log-logistic, external = exponential}}\\
\midrule
Design & Method & MSE & FP$_{\text{sus}}$ & FN$_{\text{sus}}$ & FP$_{\text{inf}}$ & FN$_{\text{inf}}$ \\
\midrule
Full Cohort (WIW obs)   & MLE    & 0.0250 & --   & --   & --   & --   \\
                        & LASSO  & 0.0712 & 0.33 & 0.00 & 0.27 & 0.00 \\
                        & ALASSO & 0.0111 & 0.00 & 0.00 & 0.01 & 0.00 \\[2pt]
Full Cohort (WIW unobs) & MLE    & 0.0371 & --   & --   & --   & --   \\
                        & LASSO  & 0.0800 & 1.19 & 0.00 & 0.80 & 0.00 \\
                        & ALASSO & 0.0178 & 0.00 & 0.00 & 0.00 & 0.00 \\[2pt]
CT (WIW obs)            & MLE    & 0.0265 & --   & --   & --   & --   \\
                        & LASSO  & 0.1055 & 0.20 & 0.00 & 0.24 & 0.00 \\
                        & ALASSO & 0.0115 & 0.00 & 0.00 & 0.00 & 0.00 \\[2pt]
CT (WIW unobs)          & MLE    & 0.0417 & --   & --   & --   & --   \\
                        & LASSO  & 0.1218 & 0.68 & 0.00 & 0.46 & 0.00 \\
                        & ALASSO & 0.0207 & 0.02 & 0.00 & 0.01 & 0.00 \\
\midrule
\multicolumn{7}{c}{\textbf{(B) Misspecified model: internal = Weibull, external = exponential}}\\
\midrule
Design & Method & MSE & FP$_{\text{sus}}$ & FN$_{\text{sus}}$ & FP$_{\text{inf}}$ & FN$_{\text{inf}}$ \\
\midrule
Full Cohort (WIW obs)   & MLE    & 0.0287 & --   & --   & --   & --   \\
                        & LASSO  & 0.0970 & 0.58 & 0.00 & 0.52 & 0.00 \\
                        & ALASSO & 0.0146 & 0.00 & 0.00 & 0.02 & 0.00 \\[2pt]
Full Cohort (WIW unobs) & MLE    & 0.0414 & --   & --   & --   & --   \\
                        & LASSO  & 0.0738 & 1.44 & 0.00 & 0.97 & 0.00 \\
                        & ALASSO & 0.0218 & 0.01 & 0.00 & 0.02 & 0.00 \\[2pt]
CT (WIW obs)            & MLE    & 0.0314 & --   & --   & --   & --   \\
                        & LASSO  & 0.1467 & 0.35 & 0.00 & 0.38 & 0.00 \\
                        & ALASSO & 0.0158 & 0.01 & 0.00 & 0.02 & 0.00 \\[2pt]
CT (WIW unobs)          & MLE    & 0.0468 & --   & --   & --   & --   \\
                        & LASSO  & 0.1454 & 0.76 & 0.00 & 0.49 & 0.00 \\
                        & ALASSO & 0.0288 & 0.05 & 0.00 & 0.02 & 0.00 \\
\midrule
\multicolumn{7}{c}{\textbf{(C) Misspecified model: internal = exponential, external = exponential}}\\
\midrule
Design & Method & MSE & FP$_{\text{sus}}$ & FN$_{\text{sus}}$ & FP$_{\text{inf}}$ & FN$_{\text{inf}}$ \\
\midrule
Full Cohort (WIW obs)   & MLE    & 0.1943 & --   & --   & --   & --   \\
                        & LASSO  & 0.0428 & 0.10 & 0.00 & 0.06 & 0.00 \\
                        & ALASSO & 0.1634 & 0.00 & 0.00 & 0.00 & 0.00 \\[2pt]
Full Cohort (WIW unobs) & MLE    & 0.1242 & --   & --   & --   & --   \\
                        & LASSO  & 0.2793 & 0.21 & 0.00 & 0.07 & 0.05 \\
                        & ALASSO & 0.0881 & 0.00 & 0.00 & 0.01 & 0.00 \\[2pt]
CT (WIW obs)            & MLE    & 0.2286 & --   & --   & --   & --   \\
                        & LASSO  & 0.0337 & 0.08 & 0.00 & 0.09 & 0.00 \\
                        & ALASSO & 0.1920 & 0.00 & 0.00 & 0.00 & 0.00 \\[2pt]
CT (WIW unobs)          & MLE    & 0.1168 & --   & --   & --   & --   \\
                        & LASSO  & 0.1575 & 0.30 & 0.00 & 0.06 & 0.00 \\
                        & ALASSO & 0.0819 & 0.00 & 0.00 & 0.02 & 0.00 \\
\bottomrule
\end{tabular}
\end{table}

Table~\ref{tab:sim_bern} summarizes estimation and variable-selection performance of the MLE, LASSO, and ALASSO across the four study designs under Bernoulli covariates.
Overall, the qualitative patterns mirror those in Table~\ref{tab:sim_mvn}, although MSE is generally larger under Bernoulli covariates.
Under correct specification (Panel A), the ALASSO achieved the smallest MSE in each design and produced sparse fits with FP counts near zero, while retaining all truly nonzero effects (FN $=0$ for both infectiousness and susceptibility).
Relative to Table~\ref{tab:sim_mvn}, the LASSO performed more poorly, exhibiting larger MSE and selecting additional noise variables, particularly when who-infected-whom was unobserved.
Under misspecification of the internal contact-interval distribution to a Weibull model (Panel B), MSE generally increased relative to Panel A, but the ALASSO remained the most accurate estimator and continued to markedly reduce false positives relative to the LASSO.
Under more severe misspecification to an exponential internal model (Panel C), the MLE exhibited pronounced loss of accuracy, and the ALASSO also showed larger MSE than in Panels A and B; however, false negative counts remained essentially zero, indicating that the true nonzero effects were rarely shrunk to zero.

\begin{table}[h!]
\scriptsize
\centering
\caption{
Coverage probabilities of nominal 95\% Wald-type confidence intervals for the four nonzero coefficients $\{\beta_{\mathrm{sus},1},\beta_{\mathrm{sus},2},\beta_{\mathrm{inf},1},\beta_{\mathrm{inf},2}\}$ under the correctly specified model (internal = log-logistic, external = exponential), where $\beta_{\mathrm{sus}}=(0.7,0.7,0,0,0)$ and $\beta_{\mathrm{inf}}=(0.7,0.7,0,0,0)$.
Panel (1) uses multivariate normal covariates with $\Sigma_{jj}=1$ and $\Sigma_{jk}=0.5$ for $j\neq k$.
Panel (2) uses independent Bernoulli$(0.5)$ covariates.
Results are based on 100 replicates.
For LASSO and ALASSO, coverage is computed conditional on selection ($\widehat{\beta}_{\lambda,k}\neq 0$); in the settings considered in Tables~\ref{tab:sim_mvn} and~\ref{tab:sim_bern}, the four nonzero coefficients were always selected (FN $=0$), so conditional and unconditional coverage coincide.
}
\label{tab:coverage}
\vspace{4pt}
\begin{tabular}{l l c c c c}
\toprule
\multicolumn{6}{c}{\textbf{(1) Multivariate normal covariates}}\\
\midrule
Design & Method & \multicolumn{4}{c}{\textbf{Coverage (95\% CI)}} \\
\cmidrule(lr){3-6}
& & $\beta_{\text{sus},1}$ & $\beta_{\text{sus},2}$ & $\beta_{\text{inf},1}$ & $\beta_{\text{inf},2}$ \\
\midrule
Full Cohort (WIW obs)        & MLE      & 0.95 & 0.98 & 0.94 & 0.93 \\
                             & LASSO    & 0.78 & 0.79 & 0.64 & 0.69 \\
                             & ALASSO & 0.96 & 0.96 & 0.97 & 0.93 \\[2pt]
Full Cohort (WIW unobs)      & MLE      & 0.94 & 0.99 & 0.92 & 0.94 \\
                             & LASSO    & 0.72 & 0.76 & 0.47 & 0.64 \\
                             & ALASSO & 0.96 & 0.98 & 0.95 & 0.94 \\[2pt]
CT (WIW obs)    & MLE      & 0.96 & 0.98 & 0.95 & 0.93 \\
                             & LASSO    & 0.72 & 0.75 & 0.62 & 0.70 \\
                             & ALASSO & 0.96 & 0.97 & 0.97 & 0.92 \\[2pt]
CT (WIW unobs)  & MLE      & 0.95 & 0.97 & 0.94 & 0.97 \\
                             & LASSO    & 0.61 & 0.66 & 0.39 & 0.53 \\
                             & ALASSO & 0.95 & 0.96 & 0.92 & 0.95 \\
\midrule
\multicolumn{6}{c}{\textbf{(2) Bernoulli covariates}}\\
\midrule
Design & Method & \multicolumn{4}{c}{\textbf{Coverage (95\% CI)}} \\
\cmidrule(lr){3-6}
& & $\beta_{\text{sus},1}$ & $\beta_{\text{sus},2}$ & $\beta_{\text{inf},1}$ & $\beta_{\text{inf},2}$ \\
\midrule
Full Cohort (WIW obs)        & MLE      & 0.91 & 0.98 & 0.97 & 0.98 \\
                             & LASSO    & 0.39 & 0.35 & 0.23 & 0.20 \\
                             & ALASSO & 0.90 & 0.96 & 0.97 & 0.98 \\[2pt]
Full Cohort (WIW unobs)      & MLE      & 0.92 & 0.96 & 0.95 & 0.99 \\
                             & LASSO    & 0.65 & 0.69 & 0.55 & 0.54 \\
                             & ALASSO & 0.92 & 0.95 & 0.92 & 0.97 \\[2pt]
CT (WIW obs)    & MLE      & 0.97 & 0.96 & 0.97 & 0.98 \\
                             & LASSO    & 0.21 & 0.22 & 0.15 & 0.14 \\
                             & ALASSO & 0.97 & 0.94 & 0.97 & 0.98 \\[2pt]
CT (WIW unobs)  & MLE      & 0.97 & 0.94 & 0.92 & 0.97 \\
                             & LASSO    & 0.52 & 0.50 & 0.17 & 0.17 \\
                             & ALASSO & 0.98 & 0.93 & 0.91 & 0.97 \\
\bottomrule
\end{tabular}
\end{table}

Table~\ref{tab:coverage} summarizes coverage of nominal 95\% Wald-type confidence intervals for the four nonzero coefficients under the correctly specified model.
For multivariate normal covariates (Panel 1), the MLE achieved near-nominal coverage across designs, ranging from 0.92 to 0.99 across coefficients.
The ALASSO was similarly well calibrated, with coverage typically between 0.92 and 0.98; for example, in the full-cohort design with observed who-infected-whom, ALASSO coverage was 0.96, 0.96, 0.97, and 0.93 for $\beta_{\mathrm{sus},1}$, $\beta_{\mathrm{sus},2}$, $\beta_{\mathrm{inf},1}$, and $\beta_{\mathrm{inf},2}$, respectively.
In contrast, LASSO coverage was substantially below nominal in every design, ranging from 0.39 to 0.79 under normal covariates; the lowest values occurred when who-infected-whom was unobserved under contact tracing, where coverage dropped to 0.61 and 0.66 for the susceptibility effects and to 0.39 and 0.53 for the infectiousness effects.
For Bernoulli covariates (Panel 2), MLE and ALASSO coverage remained close to nominal (MLE 0.91--0.99; ALASSO 0.90--0.98).
In contrast, LASSO exhibited severe undercoverage, most pronounced when WIW was observed (full cohort: 0.39, 0.35, 0.23, 0.20; contact tracing: 0.21, 0.22, 0.15, 0.14).
Under WIW unobserved, LASSO coverage remained well below nominal, ranging from 0.50--0.69 for susceptibility effects and 0.17--0.55 for infectiousness effects.
Overall, ALASSO with the active-set variance approximation achieved near-nominal Wald-type coverage for nonzero effects, whereas Wald-type inference after LASSO selection showed substantial undercoverage.

% ===============================
% Application
% ===============================

\section{Los Angeles County influenza A (H1N1) data}
\label{sec:application}

We illustrate the proposed penalized pairwise AFT framework using household influenza A (H1N1) surveillance data collected by the Los Angeles County Department of Public Health (LACDPH) during April and May 2009.
Data collection followed a standardized LACDPH protocol \citep{sugimoto2011accounting}.
First, between April 14 and May 18, nasopharyngeal swabs and aspirates were obtained from individuals presenting to LACDPH or local providers with acute febrile respiratory illness (AFRI), defined as fever $\ge 37.8^\circ$C plus at least one of cough, sore throat, or rhinorrhea, and specimens were tested for influenza using reverse-transcription polymerase chain reaction (RT-PCR).
Second, patients whose specimens tested positive for pandemic influenza A (H1N1) or influenza A of undetermined subtype were invited to complete a structured telephone interview, which collected household roster information and covariates including sex, age, and antiviral prophylaxis use; for index cases younger than 18 years, an adult proxy respondent was interviewed.
Third, symptom onset dates for AFRI episodes among household members were obtained from the initial interview and, when necessary, follow-up interviews, covering episodes occurring up to 14 days after symptom onset in the index case; all interviews were completed between April 30 and June 1.

For analysis, we assumed that all reported AFRI episodes among household members were attributable to influenza A (H1N1).
The index case in each household was assumed to be externally infected, while other household members were assumed to be at risk of infection from both infectious household contacts and the community, consistent with the contact tracing design with delayed entry considered in our simulation study.
We fixed the incubation period at 2 days, the latent period at 0 days, and the infectious period at 6 days, following assumptions adapted from \citet{yang2009transmissibility}.
Here, the incubation period is the time from infection to symptom onset, whereas the latent period is the time from infection to the onset of infectiousness.
Under these assumptions, a person $j$ with symptom onset on day $t_j^{\mathrm{sym}}$ is assumed to have been infected on day $t_j=t_j^{\mathrm{sym}}-2$ and, in the discrete-time model, to be infectious on days $t_j+1$ through $t_j+6$.
Households were enrolled when an index case presented for clinical evaluation.
We began follow-up for each non-index household member on the index case’s infection day (symptom onset minus 2 days) and ended follow-up at the member’s infection day if infected, or otherwise 14 days after the index case’s infection day.
During follow-up, members were at risk of external infection throughout and of within-household infection only during the infectious days of infected household contacts.

Covariates included sex ($\mathtt{male}$), age category ($\mathtt{adult}=1$ for ages $\ge 18$ years and $0$ otherwise), and antiviral prophylaxis.
Prophylaxis was modeled as a time-dependent covariate and was assumed to begin on the day after symptom onset in the index case.
Each pair included covariate values for the infectious individual (\(\mathtt{male\_inf}\), \(\mathtt{adult\_inf}\), \(\mathtt{proph\_inf}\)) as well as for the susceptible individual (\(\mathtt{male\_sus}\), \(\mathtt{adult\_sus}\), \(\mathtt{proph\_sus}\)).
For external exposure pairs ($i=0$), infectiousness covariates were set to zero.
We considered exponential, Weibull, and log-logistic models for the internal and external contact-interval distributions.

\subsection{Data analysis results}
The LACDPH household data set comprised 299 individuals from 58 households.
A total of 99 probable influenza infections were identified, including 62 index cases; four households had co-primary cases with symptom onset on the same day.
Missing covariate information was present for sex (3 individuals), age (4 individuals), and antiviral prophylaxis (56 individuals), for a total of 62 individuals with at least one missing covariate.
These 62 individuals came from 17 households and accounted for 36 infections, including 19 index cases.
Because we assumed that all household members could infect or be infected by other household members, we excluded an entire household if any member had missing data.
The resulting complete-case data set included 41 households with 63 infections, of which 43 were index cases.
Using this complete-case data set, we fit main-effects models including all six covariates under each of the nine combinations of internal and external contact-interval distributions (exponential, Weibull, and log-logistic for each component), estimating parameters by maximum likelihood \citep{sharker2024pairwise}.
As reported in Table~5 of \citet{sharker2024pairwise}, the lowest AIC was achieved by the model with exponential internal contact intervals and log-logistic external contact intervals (AIC $=207.66$).

We fit the exponential internal and log-logistic external transmission model using the MLE and the penalized estimators (LASSO and ALASSO), selecting $\lambda$ by BIC over 100 values in $[0.01,1]$ (Table~\ref{tab:mle_lasso_alasso}).
Consistent with the simulation study, ALASSO produced the sparsest fit, retaining the age effects on infectiousness and susceptibility and the prophylaxis effect on susceptibility, while shrinking the sex effects and the prophylaxis effect on infectiousness to zero.
Across estimators, the age effects were directionally consistent with higher infectiousness and lower susceptibility in adults than in children, although uncertainty remained substantial.
The prophylaxis effect on susceptibility was more clearly supported by the data: under ALASSO, the corresponding rate ratio was $\exp(-0.950)=0.39$ (approximate 95\% CI $0.17$ to $0.87$), suggesting reduced infection risk among susceptibles receiving antiviral prophylaxis.
In contrast, the unpenalized MLE yielded very imprecise estimates for some effects, most notably prophylaxis on infectiousness, whereas the penalized fits improved numerical stability by shrinking weakly supported effects toward zero.

\begin{table}[h!]
\centering
\scriptsize
\caption{
MLE, LASSO, and adaptive LASSO estimates from the transmission model with exponential internal and log-logistic external contact intervals.
Coefficients shrunk to zero by the LASSO or adaptive LASSO penalty are shown as ``--''.
Tuning parameters were selected by BIC over 100 values in $[0.01,1]$, yielding $\lambda$ for LASSO = 0.57 and $\lambda$ for adaptive LASSO = 0.48 with $\eta=1.75$.
Standard errors (SE) are reported for MLE and for the active set in the penalized fits.
}
\vspace{1em}
\label{tab:mle_lasso_alasso}
\begin{tabular}{lrrrrrr}
\toprule
\textbf{Parameter}
& \textbf{MLE}
& \textbf{(SE)}
& \textbf{LASSO}
& \textbf{(SE)}
& \textbf{ALASSO}
& \textbf{(SE)} \\
\midrule
Intercept
& $-4.758$ & $(1.507)$
& $-4.469$ & $(0.610)$
& $-4.734$ & $(0.720)$ \\

Inf: Male
& $-0.693$ & $(1.299)$
& --       & --
& --       & -- \\

Sus: Male
& $-0.327$ & $(0.334)$
& $-0.235$ & $(0.254)$
& --       & -- \\

Inf: Adult
& $1.677$  & $(1.563)$
& $0.825$  & $(0.569)$
& $1.034$  & $(0.733)$ \\

Sus: Adult
& $-0.465$ & $(0.334)$
& $-0.422$ & $(0.281)$
& $-0.389$ & $(0.260)$ \\

Inf: Prophylaxis
& $-4.206$ & $(30.044)$
& --       & --
& --       & -- \\

Sus: Prophylaxis
& $-1.023$ & $(0.489)$
& $-0.912$ & $(0.393)$
& $-0.950$ & $(0.416)$ \\

External rate (log)
& $-3.925$ & $(0.335)$
& $-4.029$ & $(0.346)$
& $-4.152$ & $(0.365)$ \\

External shape (log)
& $0.854$  & $(0.403)$
& $0.870$  & $(0.409)$
& $0.843$  & $(0.437)$ \\
\bottomrule
\end{tabular}
\end{table}

\begin{table}[h!]
\centering
\small
\caption{
Predicted household secondary attack risks with 95\% Wald confidence intervals from ALASSO estimates under the selected model (exponential internal, log-logistic external contact intervals).
}
\vspace{1em}
\label{tab:sar}
\begin{tabular}{llrc}
\toprule
\textbf{From} & \textbf{To} & \textbf{Estimate} & \textbf{95\% confidence interval} \\
\midrule
\multirow{2}{*}{Child}
& Child untreated
& 5.1\% & (1.3\%, 19.5\%) \\
& Child on prophylaxis
& 2.0\% & (0.5\%, 8.6\%) \\
\midrule
\multirow{2}{*}{Child}
& Adult untreated
& 3.5\% & (0.9\%, 12.9\%) \\
& Adult on prophylaxis
& 1.4\% & (0.3\%, 5.7\%) \\
\midrule
\multirow{2}{*}{Adult}
& Child untreated
& 13.8\% & (5.9\%, 30.6\%) \\
& Child on prophylaxis
& 5.6\% & (2.1\%, 14.3\%) \\
\midrule
\multirow{2}{*}{Adult}
& Adult untreated
& 9.6\% & (4.2\%, 21.1\%) \\
& Adult on prophylaxis
& 3.8\% & (1.4\%, 9.9\%) \\
\bottomrule
\end{tabular}
\end{table}

Table~\ref{tab:sar} reports predicted household secondary attack risks (SAR) stratified by the age of the infectious individual, the age of the susceptible individual, and antiviral prophylaxis in the susceptible individual, using the ALASSO coefficient estimates.
Predicted risks are higher when the infectious individual is an adult and lower when the susceptible individual is an adult, consistent with higher infectiousness and lower susceptibility among adults in this data set.
Across strata, antiviral prophylaxis in the susceptible individual is associated with a marked reduction in predicted risk; for example, when the infectious individual is an adult and the susceptible individual is a child, the predicted SAR decreases from 13.8\% without prophylaxis to 5.6\% with prophylaxis.
We report 95\% Wald confidence intervals computed from the ALASSO coefficient estimates and the corresponding active set covariance matrix.

\section{Discussion}
We proposed a penalized pairwise accelerated failure time (AFT) model for contact intervals that jointly represents within-population transmission and infection from external sources.
A coordinate descent algorithm enables scalable estimation, and simulations show that the adaptive LASSO yields sparse fits with improved estimation accuracy and near-nominal Wald-type coverage for selected nonzero coefficients.
In the Los Angeles County influenza A (H1N1) application, the adaptive LASSO produced an interpretable model and improved precision for selected effects, including antiviral prophylaxis on susceptibility.

Two important directions remain.
First, our oracle property result is established only when who infected whom is observed.
When who infected whom is not observed, the likelihood is based on aggregated hazards and no longer decomposes additively over ordered pairs, and we have not yet proved an oracle property in this setting, although simulations suggest good performance.
Second, our current approach is fully parametric, so misspecification of the contact interval distribution can compromise accuracy.
While semiparametric AFT methods \citep{buckley1979linear, wei1990linear, jin2003rank, jin2006least} are well developed, analogous pairwise survival models that leave the baseline contact interval distribution unspecified remain largely unexplored.
Such extensions could improve robustness to baseline misspecification while preserving interpretable covariate effects, and developing penalized estimation in this setting is an important direction for future work.

\section*{Funding}
EK and JL were supported by National Institute of Allergy and Infectious Diseases (NIAID) grant R01 AI184507.
The content is solely the responsibility of the authors and does not represent the official views or policies of NIAID or the National Institutes of Health.

\section*{Acknowledgements}
We are grateful to the Los Angeles County Department of Public Health for allowing the use of their data.

\section*{Conflicts of Interest}
No potential conflict of interest was reported by the author(s).

\section*{Supplementary Material}
The file \texttt{Supplementary.zip} contains all code and data used in this paper, including \texttt{README.txt}.
The \texttt{README.txt} file provides an overview of the directory structure and step-by-step instructions for running the R scripts to reproduce the simulation study and the real-data analyses.

% ===============================
% References
% ===============================
\bibliographystyle{chicago}
\bibliography{ref}
% ===============================
% Appendix
% ===============================

\appendix

\section*{Appendix}

\subsection*{Coordinate Descent Algorithm}
We outline the derivation of the update rule in Equation~\eqref{eq:update}.
Fix $k\in\{1,\ldots,p\}$ and hold all other parameters $\{\beta_j:j\neq k\}$ and $\bm\psi$ fixed at their current values.
Let
\[
f_k(z)
=
-\ell(\beta_1,\ldots,\beta_{k-1},z,\beta_{k+1},\ldots,\beta_{p};\bm\psi)
\]
denote the one-dimensional slice of the \emph{unpenalized} negative log likelihood, and consider the local penalized objective
\[
Q_k(\beta_k+\Delta)
=
f_k(\beta_k+\Delta)
+\lambda w_k\lvert \beta_k+\Delta\rvert .
\]
A second-order expansion of $f_k$ around the current iterate yields
\[
f_k(\beta_k+\Delta)
\approx
f_k(\beta_k)
+ g_k\,\Delta
+ \tfrac12 H_{kk}\,\Delta^2,
\]
where
\[
g_k
=
\left.\frac{\partial}{\partial z}\, f_k(z)\right|_{z=\beta_k},
\qquad
H_{kk}
=
\left.\frac{\partial^2}{\partial z^2}\, f_k(z)\right|_{z=\beta_k},
\qquad
H_{kk}>0.
\]
Ignoring constants that do not depend on $\Delta$, we minimize
\[
\phi(\Delta)
=
\tfrac12 H_{kk}\Delta^2 + g_k\Delta + \lambda w_k\lvert \beta_k+\Delta\rvert .
\]
Since $H_{kk}>0$, $\phi(\Delta)$ is convex, and the necessary and sufficient optimality condition is
\[
0 \in \partial \phi(\Delta)
=
H_{kk}\Delta + g_k + \lambda w_k\,\partial\lvert \beta_k+\Delta\rvert ,
\]
where $\partial|\cdot|$ denotes the subgradient of the absolute-value function:
\[
\partial|x|=
\begin{cases}
\{+1\}, & x>0,\\[3pt]
[-1,+1], & x=0,\\[3pt]
\{-1\}, & x<0.
\end{cases}
\]

\noindent
\textbf{Case 1:} $\beta_k+\Delta>0$.
Then $\partial|\beta_k+\Delta|=\{+1\}$ and
\[
0 = H_{kk}\Delta + g_k + \lambda w_k
\quad\Rightarrow\quad
\beta_k^{\mathrm{new}}=\beta_k-\frac{g_k+\lambda w_k}{H_{kk}},
\]
valid if $\beta_k-\frac{g_k}{H_{kk}}>\frac{\lambda w_k}{H_{kk}}$.

\noindent
\textbf{Case 2:} $\beta_k+\Delta<0$.
Then $\partial|\beta_k+\Delta|=\{-1\}$ and
\[
0 = H_{kk}\Delta + g_k - \lambda w_k
\quad\Rightarrow\quad
\beta_k^{\mathrm{new}}=\beta_k-\frac{g_k-\lambda w_k}{H_{kk}},
\]
valid if $\beta_k-\frac{g_k}{H_{kk}}<-\frac{\lambda w_k}{H_{kk}}$.

\noindent
\textbf{Case 3:} $\beta_k+\Delta=0$.
Then $\partial|\beta_k+\Delta|=[-1,+1]$ and the optimality condition
\[
0 \in -H_{kk}\beta_k + g_k + \lambda w_k[-1,+1]
\]
holds when
\[
\left|\beta_k-\frac{g_k}{H_{kk}}\right|
\le
\frac{\lambda w_k}{H_{kk}}.
\]
Combining the three cases yields the soft-thresholding update in Equation~\eqref{eq:update}.

\subsection*{Oracle property}

We use the notation and assumptions in Section~\ref{sec:asymptotic}.
The proof closely follows Appendix~1 of \citet{zhang2007adaptive}, with the profile score and information normalized by $m$.
The local root-$m$ reparameterization follows the standard LASSO asymptotic argument of \citet{knight2000asymptotics}.
The rate calculation for the general weights $|\widetilde\beta_{m,k}|^{-\eta}$ follows the same argument as in \citet{zou2006adaptive}.

\subsubsection*{Proof of Theorem~\ref{thm:rootm}}
\begin{proof}
Consider the $C$-ball
\[
\mathcal B_m(C)
=
\left\{
\bm\beta:
\bm\beta=\bm\beta_0+m^{-1/2}\bm u,\
\|\bm u\|\le C
\right\},
\]
and denote its boundary by $\partial\mathcal B_m(C)$.
It is sufficient to show that, for every $\epsilon>0$, there exists a sufficiently large constant $C$ such that
\begin{equation}
\Pr\left\{
\sup_{\bm\beta\in\partial\mathcal B_m(C)}
\mathcal Q_m(\bm\beta)
<
\mathcal Q_m(\bm\beta_0)
\right\}>1-\epsilon
\label{eq:rootm_boundary}
\end{equation}
for all sufficiently large $m$.

By \eqref{eq:profile_score} and \eqref{eq:profile_information},
\[
\bm S_m=O_p(1),
\qquad
-\frac{1}{m}
\ddot\ell_{\text{profile},m}(\bm\beta)
=
\bm\Sigma+o_p(1)
\]
uniformly for $\bm\beta\in\mathcal B_m(C)$.
Hence, for $\bm\beta=\bm\beta_0+m^{-1/2}\bm u$, a second-order Taylor expansion gives
\begin{equation}
\ell_{\text{profile},m}
\left(\bm\beta_0+m^{-1/2}\bm u\right)
-
\ell_{\text{profile},m}(\bm\beta_0)
=
\bm u^\top\bm S_m
-
\frac12
\bm u^\top\{\bm\Sigma+o_p(1)\}\bm u
\label{eq:profile_local_expansion}
\end{equation}
uniformly for $\|\bm u\|\le C$.

For $k\in\mathcal S_0$, root-$m$ consistency of the initial estimator implies
\[
w_{m,k}
=
|\widetilde\beta_{m,k}|^{-\eta}
=
|\beta_{0,k}|^{-\eta}+o_p(1)
=
O_p(1).
\]
For $\|\bm u\|=C$, the reverse triangle inequality therefore yields
\begin{align}
&\mathcal Q_m
\left(\bm\beta_0+m^{-1/2}\bm u\right)
-
\mathcal Q_m(\bm\beta_0)
\notag\\
&=
\ell_{\text{profile},m}
\left(\bm\beta_0+m^{-1/2}\bm u\right)
-
\ell_{\text{profile},m}(\bm\beta_0)
\notag\\
&\quad
-\lambda_m
\sum_{k\in\mathcal S_0}w_{m,k}
\left(
\left|\beta_{0,k}+m^{-1/2}u_k\right|
-
|\beta_{0,k}|
\right)
-
\lambda_m m^{-1/2}
\sum_{k\notin\mathcal S_0}w_{m,k}|u_k|
\notag\\
&\le
\ell_{\text{profile},m}
\left(\bm\beta_0+m^{-1/2}\bm u\right)
-
\ell_{\text{profile},m}(\bm\beta_0)
+
\lambda_m
\sum_{k\in\mathcal S_0}w_{m,k}
\left|
\left|\beta_{0,k}+m^{-1/2}u_k\right|
-
|\beta_{0,k}|
\right|
\notag\\
&\le
\ell_{\text{profile},m}
\left(\bm\beta_0+m^{-1/2}\bm u\right)
-
\ell_{\text{profile},m}(\bm\beta_0)
+
\lambda_m m^{-1/2}
\sum_{k\in\mathcal S_0}w_{m,k}|u_k|
\notag\\
&=
-\frac12
\bm u^\top\{\bm\Sigma+o_p(1)\}\bm u
+
\bm u^\top\bm S_m
+
\lambda_m m^{-1/2}
\sum_{k\in\mathcal S_0}w_{m,k}|u_k|
\notag\\
&\le
-\kappa C^2
+
C\,O_p(1)
+
C\frac{\lambda_m}{\sqrt m}O_p(1),
\label{eq:rootm_inequality_chain}
\end{align}
where $\kappa>0$ is a fixed constant because $\bm\Sigma$ is positive definite.
Under $\lambda_m/\sqrt m=O(1)$, the negative quadratic term is of order $-C^2$, whereas the other two terms are of order $C\,O_p(1)$.
Thus, by choosing $C$ sufficiently large, \eqref{eq:rootm_boundary} holds.

Since $\bm\beta_0$ is an interior point, $\mathcal B_m(C)$ is contained in the parameter space for all sufficiently large $m$.
Since $\mathcal Q_m$ is continuous, it attains a maximum over the compact set $\mathcal B_m(C)$.
By \eqref{eq:rootm_boundary}, a maximizer over the ball cannot lie on its boundary.
It is therefore a local maximizer and satisfies
\[
\left\|
\widehat{\bm\beta}_{\lambda,m}-\bm\beta_0
\right\|
=
O_p(m^{-1/2}).
\]
This proves Theorem~\ref{thm:rootm}.
\end{proof}

\subsubsection*{Proof of Theorem~\ref{thm:oracle}}
\begin{proof}
\noindent\textit{(i) Sparsity.}
Fix $k\notin\mathcal S_0$.
For $\bm\beta=\bm\beta_0+m^{-1/2}\bm u$ with $\|\bm u\|\le C$, Taylor expansion and \eqref{eq:profile_score}--\eqref{eq:profile_information} give
\[
\frac{1}{\sqrt m}
\dot\ell_{\text{profile},m,k}(\bm\beta)
=
S_{m,k}
-
\{\bm\Sigma_{k\cdot}+o_p(1)\}\bm u
=
O_p(1)
\]
uniformly over this neighborhood, where $\bm\Sigma_{k\cdot}$ denotes the $k$th row of $\bm\Sigma$.
Since $\beta_{0,k}=0$ and the initial estimator is root-$m$ consistent,
\[
Z_{m,k}
=
\sqrt m\,\widetilde\beta_{m,k}
=
O_p(1),
\]
and hence
\[
\frac{\lambda_mw_{m,k}}{\sqrt m}
=
\lambda_m m^{(\eta-1)/2}|Z_{m,k}|^{-\eta}
\to_p\infty.
\]
For $\beta_k\neq0$, it follows that
\begin{align}
\frac{1}{\sqrt m}
\frac{\partial\mathcal Q_m(\bm\beta)}{\partial\beta_k}
&=
\frac{1}{\sqrt m}
\dot\ell_{\text{profile},m,k}(\bm\beta)
-
\frac{\lambda_mw_{m,k}}{\sqrt m}
\operatorname{sign}(\beta_k)
\notag\\
&=
O_p(1)
-
\left\{
\lambda_m m^{(\eta-1)/2}
|Z_{m,k}|^{-\eta}
\right\}
\operatorname{sign}(\beta_k).
\label{eq:inactive_derivative_sign}
\end{align}
Because the quantity in braces diverges in probability, the sign of \eqref{eq:inactive_derivative_sign} is opposite to the sign of $\beta_k$ with probability tending to one.
Therefore a root-$m$ consistent local maximizer must satisfy
\[
\widehat\beta_{\lambda,m,k}=0.
\]
Since $p$ is fixed, this holds simultaneously for every $k\notin\mathcal S_0$, and therefore
\[
\Pr\left(
\widehat{\bm\beta}_{\lambda,m,\mathcal S_0^c}=\bm 0
\right)\to1.
\]

\noindent\textit{(ii) Asymptotic normality.}
For $k\in\mathcal S_0$, root-$m$ consistency of the initial estimator gives
\[
w_{m,k}
=
|\beta_{0,k}|^{-\eta}+o_p(1)
=
O_p(1).
\]
Root-$m$ consistency of $\widehat{\bm\beta}_{\lambda,m}$ and $\beta_{0,k}\neq0$ also imply
\[
\operatorname{sign}
(\widehat\beta_{\lambda,m,k})
=
\operatorname{sign}(\beta_{0,k})
\]
with probability tending to one.
Define the active penalty-gradient vector
\[
\bm b_m
=
\left(
w_{m,k}\operatorname{sign}(\beta_{0,k})
:
k\in\mathcal S_0
\right)^\top.
\]
Because $p$ is fixed and $w_{m,k}=O_p(1)$ for every $k\in\mathcal S_0$, we have $\bm b_m=O_p(1)$.

By part (i), with probability tending to one,
\[
\widehat{\bm\beta}_{\lambda,m,\mathcal S_0^c}
-
\bm\beta_{0,\mathcal S_0^c}
=
\bm 0.
\]
Together with sign consistency on $\mathcal S_0$, the active first-order conditions can therefore be written as
\[
\bm 0
=
\dot\ell_{\text{profile},m,\mathcal S_0}
(\widehat{\bm\beta}_{\lambda,m})
-
\lambda_m\bm b_m.
\]
Expanding the active profile score about $\bm\beta_0$ gives
\[
\bm 0
=
\sqrt m\,\bm S_{m,\mathcal S_0}
-
m\{\bm\Sigma_{\mathcal S_0}+o_p(1)\}
\left(
\widehat{\bm\beta}_{\lambda,m,\mathcal S_0}
-
\bm\beta_{0,\mathcal S_0}
\right)
-
\lambda_m\bm b_m.
\]
The cross-Hessian term does not appear because its displacement factor $\widehat{\bm\beta}_{\lambda,m,\mathcal S_0^c}-\bm\beta_{0,\mathcal S_0^c}$ is zero by part (i).
Consequently,
\[
\sqrt m
\left(
\widehat{\bm\beta}_{\lambda,m,\mathcal S_0}
-
\bm\beta_{0,\mathcal S_0}
\right)
=
\{\bm\Sigma_{\mathcal S_0}+o_p(1)\}^{-1}
\left\{
\bm S_{m,\mathcal S_0}
-
\frac{\lambda_m}{\sqrt m}\bm b_m
\right\}.
\]
Since $\lambda_m/\sqrt m\to0$ and $\bm b_m=O_p(1)$, the penalty term is $o_p(1)$.
Thus, by \eqref{eq:profile_score} and Slutsky's theorem,
\[
\sqrt m
\left(
\widehat{\bm\beta}_{\lambda,m,\mathcal S_0}
-
\bm\beta_{0,\mathcal S_0}
\right)
\to_d
N\left(\bm 0,\bm\Sigma_{\mathcal S_0}^{-1}\right).
\]
This proves Theorem~\ref{thm:oracle}.
\end{proof}

\subsection*{Approximate variance}
\label{app:lqa_cov}
We provide the derivation of the covariance estimator in Section~\ref{sec:alasso_var} under the local quadratic approximation (LQA) of the adaptive LASSO penalty \citep{fan2001variable,zou2006adaptive,zhang2007adaptive}.

Using the LQA and the block-diagonal matrix $\bm A_{(r)}$ defined in Section~\ref{sec:alasso_var}, the penalty is approximated by $\lambda\bm\theta^\top\bm A_{(r)}\bm\theta/2$.
Accordingly, define the surrogate objective
\[
\widetilde Q_{(r)}(\bm\theta)
=
\ell(\bm\theta)
-
\frac{\lambda}{2}\bm\theta^\top \bm A_{(r)}\bm\theta.
\]
Using the score $\bm U$ and observed information $\bm H$ defined in Section~\ref{sec:alasso_var}, the gradient and Hessian of $\widetilde Q_{(r)}$ are
\[
\nabla \widetilde Q_{(r)}(\bm\theta)
=
\bm U(\bm\theta)-\lambda \bm A_{(r)}\bm\theta,
\qquad
\nabla^2 \widetilde Q_{(r)}(\bm\theta)
=
-\bm H(\bm\theta)-\lambda \bm A_{(r)}.
\]
Applying one Newton step to maximize $\widetilde Q_{(r)}$ gives
\begin{align*}
\widehat{\bm\theta}^{(r+1)}
&=
\widehat{\bm\theta}^{(r)}
-
\left\{\nabla^2 \widetilde Q_{(r)}\big(\widehat{\bm\theta}^{(r)}\big)\right\}^{-1}
\nabla \widetilde Q_{(r)}\big(\widehat{\bm\theta}^{(r)}\big)
\\
&=
\widehat{\bm\theta}^{(r)}
+
\left\{\bm H\big(\widehat{\bm\theta}^{(r)}\big)+\lambda \bm A_{(r)}\right\}^{-1}
\left\{\bm U\big(\widehat{\bm\theta}^{(r)}\big)-\lambda \bm A_{(r)}\widehat{\bm\theta}^{(r)}\right\},
\end{align*}
since $\{-\bm H-\lambda\bm A\}^{-1}=-(\bm H+\lambda\bm A)^{-1}$.
Multiplying both sides by $\bm H\big(\widehat{\bm\theta}^{(r)}\big)+\lambda \bm A_{(r)}$ yields
\[
\left\{\bm H\big(\widehat{\bm\theta}^{(r)}\big)+\lambda \bm A_{(r)}\right\}
\left(\widehat{\bm\theta}^{(r+1)}-\widehat{\bm\theta}^{(r)}\right)
=
\bm U\big(\widehat{\bm\theta}^{(r)}\big)-\lambda \bm A_{(r)}\widehat{\bm\theta}^{(r)}.
\]
At convergence, write $\widehat{\bm\theta}_\lambda$ and let $\bm A_\lambda$ denote $\bm A_{(r)}$ evaluated at the final iterate.
Under the LQA, $\widehat{\bm\theta}_\lambda$ approximately solves the penalized score equation
\[
\bm g(\bm\theta)
=
\bm U(\bm\theta)-\lambda \bm A_\lambda \bm\theta
=
\bm 0.
\]
Let $\bm\theta_0$ denote the target parameter.
A first-order Taylor expansion gives
\[
\bm 0
=
\bm g(\widehat{\bm\theta}_\lambda)
\approx
\bm g(\bm\theta_0)
+
\bm G(\bm\theta_0)\big(\widehat{\bm\theta}_\lambda-\bm\theta_0\big),
\]
where
\[
\bm G(\bm\theta)
=
\frac{\partial \bm g(\bm\theta)}{\partial\bm\theta^\top}
=
\frac{\partial \bm U(\bm\theta)}{\partial\bm\theta^\top}-\lambda \bm A_\lambda
=
\nabla^2 \ell(\bm\theta)-\lambda \bm A_\lambda
=
-\bm H(\bm\theta)-\lambda \bm A_\lambda.
\]
Hence
\[
\widehat{\bm\theta}_\lambda-\bm\theta_0
\approx
\left\{\bm H(\bm\theta_0)+\lambda \bm A_\lambda\right\}^{-1}
\bm g(\bm\theta_0).
\]
Taking covariance on both sides yields the sandwich form
\begin{equation}
\mathrm{cov}(\widehat{\bm\theta}_\lambda)
\approx
\left\{\bm H(\bm\theta_0)+\lambda \bm A_\lambda\right\}^{-1}
\,
\mathrm{cov}\left\{\bm g(\bm\theta_0)\right\}
\,
\left\{\bm H(\bm\theta_0)+\lambda \bm A_\lambda\right\}^{-1}.
\label{eq:app_sandwich_g_theta}
\end{equation}
We plug in $\mathrm{cov}\{\bm g(\bm\theta_0)\}=\mathrm{cov}\{\bm U(\bm\theta_0)\}$ and use the standard likelihood identity $\mathrm{cov}\{\bm U(\bm\theta_0)\}=\E\{\bm H(\bm\theta_0)\}$ under correct specification.
With the plug-in substitutions $\bm\theta_0\mapsto \widehat{\bm\theta}_\lambda$ and $\E\{\bm H(\bm\theta_0)\}\mapsto \bm H(\widehat{\bm\theta}_\lambda)$, we obtain the covariance estimator in \eqref{eq:alasso_cov_theta}.
Restricting that estimator to $\mathcal A=\mathcal U\cup\mathcal S$ gives the active-set covariance estimator described in Section~\ref{sec:alasso_var}.

\end{document}